\documentclass[submission,copyright]{eptcs}
\providecommand{\event}{NCMA 2023} 

\usepackage[utf8]{inputenc}
\usepackage{underscore}

\usepackage{iftex}

\ifpdf
  \usepackage{underscore}         
  \usepackage[T1]{fontenc}        
\else
  \usepackage{breakurl}           
\fi

\newcommand{\eps}{\ensuremath{\varepsilon}}

\newcommand{\sig}{\ensuremath{\Sigma}}

\newcommand{\changeB}[1]{#1}
\newcommand{\changeM}[1]{#1}

\makeatletter
\long\def\ifnodedefined#1#2#3{%
    \@ifundefined{pgf@sh@ns@#1}{#3}{#2}%
}
\makeatother

\usepackage{amsmath}
\usepackage{amssymb}
\usepackage{slashed}
\usepackage{wrapfig}
\usepackage{MnSymbol} 
\usepackage{stmaryrd}
\usepackage{mathtools}
\usepackage{moreverb}
\usepackage{amsthm}

\usepackage{tikz}
\usepackage{tikzpagenodes}
\usetikzlibrary{positioning,automata,shapes,calc,patterns,patterns.meta,decorations.pathreplacing,arrows,arrows.meta}

\tikzset{>=latex'}

\renewcommand{\emptyset}{\varnothing}

\pgfdeclarelayer{background}
\pgfdeclarelayer{doublebackground}
\pgfdeclarelayer{triplebackground}
\pgfsetlayers{triplebackground,doublebackground,background,main}

\makeatletter
\newcommand\footnoteref[1]{\protected@xdef\@thefnmark{\ref{#1}}\@footnotemark}
\makeatother

\newcommand{\tok}{\ensuremath{\mathord{\,\raisebox{0.5pt}{\ensuremath{\wr}}\,}}}
\newcommand{\Tk}{\ensuremath{\mathbb{T}}}
\newcommand{\TkEmpty}{\ensuremath{\Tk^{\emptyset}}}
\newcommand{\eqed}{\hfill\ensuremath{\diamond}}

\DeclareMathOperator{\vocab}{vocab}
\DeclareMathOperator{\maxvocab}{maxvocab}
\DeclareMathOperator{\maxtok}{maxtok}

\newtheorem{theorem}{Theorem}
\newtheorem{lemma}{Lemma}
\newtheorem{corollary}{Corollary}

\theoremstyle{definition}
\newtheorem{definition}{Definition}

\newtheorem{algorithm}{Algorithm}
\theoremstyle{remark}
\newtheorem{example}{Example}
\newtheorem{remark}{Remark}

\title{Formalizing BPE Tokenization}
\author{Martin Berglund
  \institute{Department of Computing Science\\Umeå University\\Umeå, Sweden}
  \email{mbe@cs.umu.se}
\and
  Brink van der Merwe
  \institute{Department of Computer Science\\Stellenbosch University\\ Stellenbosch, South Africa}
  \email{abvdm@cs.sun.ac.za}
}

\def\authorrunning{M. Berglund \& B. van der Merwe}
\def\titlerunning{Formalizing BPE Tokenization}

\newcommand{\tokenbox}[1]{%
  \tikz[baseline] {%
    \node[outer sep=0,inner sep=0,anchor=base west] (text) {\hskip.7pt #1\hskip1.5pt\vphantom{$\hat{I}$}};%
    \begin{pgfonlayer}{background}%
      \fill[white] (-.1em,-.2em) rectangle (text.north east);%
      \draw[overlay,very thick,red!20!white] (-.1em,-.2em) rectangle (text.north east);%
    \end{pgfonlayer}
  }%
}

\begin{document}

\maketitle

\begin{abstract}
  In  this paper, we formalize practical byte pair encoding tokenization as it is used in large language models and other NLP systems, in particular we formally define and investigate the semantics of the SentencePiece and HuggingFace tokenizers, in particular how they relate to each other, depending on how the tokenization rules are constructed. Beyond this we consider how tokenization can be performed in an incremental fashion, as well as doing it left-to-right using an amount of memory constant in the length of the string, enabling e.g.\ using a finite state string-to-string transducer.
\end{abstract}

\section{Introduction}

Many modern NLP systems, for example large language models such as the GPT models which underpin services like ChatGPT~\cite{chatgpt}, operate on a \emph{tokenization} of text. This tokenization defines an alphabet of symbols (in the formal languages sense) which include as many common words and fragments of words as possible. For example, the OpenAI GPT-2 model has an alphabet size (a ``vocabulary'' in their terminology) of~50,257 tokens, which is enough to turn the sentence ``taking a ride on a boat'' into\footnote{We are for the purposes of this example ignoring some details, involving whitespace and the string beginning and end.} the token sequence ``\tokenbox{taking} \tokenbox{a} \tokenbox{ride} \tokenbox{on} \tokenbox{a} \tokenbox{boat}'', as all words are common enough to be in the alphabet, but ``partaking in a nautical excursion aboard a vessel'' is tokenized as ``\tokenbox{part}\tokenbox{aking} \tokenbox{in} \tokenbox{a} \tokenbox{n}\tokenbox{autical} \tokenbox{exc}\tokenbox{ursion} \tokenbox{ab}\tokenbox{oard} \tokenbox{a} \tokenbox{ves}\tokenbox{sel}'', as the words are uncommon, but fragments of the words are common enough. This alphabet is then more semantically rich. For example, forming unusual plurals (turning ``\tokenbox{earths}'' into ``\tokenbox{earth}\tokenbox{s}''), or making a noun into a ``non-existing'' verb (i.e.\ turning ``\tokenbox{verb}'' into ``\tokenbox{verb}\tokenbox{ing}''), retaining the informative root word, but also it generally makes the model more robust to misspellings and other minor transformations of the text~\cite{bpe}.

One common way of performing this tokenization is by byte pair encoding~\cite{bpe} (BPE), used by OpenAI GPT models, and e.g.\ the recent Swedish GPT-SW3 model~\cite{gpt-sw3-tokenizer}, which uses the Google tokenizer implementation SentencePiece~\cite{sentencepiece}. BPE operates similar to a compression technique, with a dictionary of token merges constructed greedily maximizing the number of tokens that get merged in a training set (see Remark~\ref{remark:bpe-construction} for a sketch of the procedure).

\changeM{
  There are other methods for performing this type of tokenization, e.g.\ the unigram language model also implemented in SentencePiece~\cite{sentencepiece}, where tokens are individually weighted. We do not consider this case here. BPE tokenization can be contrasted to lexical analysis~\cite{dragon-book}, as lexical analysis (as exhibited in e.g.\ the POSIX tool \texttt{lex}) differs in that the rules are typically authored by hand, and break the string into an infinite language of tokens divided among a constant set of categories. Consider for example extracting arbitrarily long identifiers and string constants when parsing programming languages. BPE tokenization can also be viewed as one case of text segmentation in natural language processing (see e.g.~\cite{text-tiling}), which covers e.g.\ breaking a text into topics, sentences, or words. These text segmentation algorithms have commonly been at least partially supervised or authored, where the tokenization algorithms considered here are designed for language-agnostic unsupervised learning.
}

The way the tokenization procedure is defined and implemented in common \changeB{tools~\cite{huggingface-gpt-2-py,sentencepiece}} is essentially \emph{global:} the highest priority rule that can be applied to the text is, no matter where in the text this application would happen. This is not \emph{usually} a problem, as the text is \emph{pretokenized} by splitting it on whitespace before applying the main tokenization procedure. That is sufficient to make tokenization algorithm behavior irrelevant in general, however:
\begin{itemize}
\item Some natural languages simply do not use interword whitespace, such as writing systems for Thai, Chinese and Japanese.
\item In many artificial (e.g.\ programming) languages whitespace is common but not required, for example ``minified'' code is very common~\cite{minification}, where all unnecessary whitespace is removed to reduce file sizes. Language models are commonly trained at least partially on code, e.g.\ the GPT-SW3 dataset contains 9.5\% code in various programming languages~\cite{gpt-sw3}.
\item In some cases there is whitespace, but pretokenizing using it does not produce the best tokens. E.g.\ in the SQL query language ``LEFT OUTER JOIN'' is a single concept, and would ideally be a single token.
\item Even in cases where the pretokenization normally works well, such as for English, a hypothetical system relying heavily (more so than any currently existing system) on the pretokenization creating short text fragments may then be vulnerable to a denial of service attack. Compare to e.g.~\cite{redos} on such attacks for regular expression matchers.
\end{itemize}

As such, as these systems find their way into broader and more complex use, it becomes interesting to investigate robust algorithms for operating on tokenizations. We investigate the following questions.
\begin{itemize}
  \item Can we perform a BPE tokenization online: observing a stream of text can we output a stream of tokens using only limited memory and computation?
  \item If we have a tokenization of a large text, and the text is modified in a localized way, can we compute a localized update of the tokenization?
\end{itemize}

The answers to all of these questions are interconnected, but first we need firm definitions of the semantics we are considering.
The outline of this paper is as follows. After introducing our notation, we define (formally) SentencePiece and HuggingFace tokenizers. Then we consider how to do tokenization in a streaming fashion. This is followed by a short section outlining our envisioned future work.

\section{Notation}

An alphabet $\sig$ is a finite set of symbols. Let $\sig^*$ denote all strings over the alphabet $\sig$, including the empty string $\eps$, and $\sig^+=\sig^*\setminus \{\eps\}$. A sequence of non-empty strings is a \emph{tokenization}, e.g.\ for $u_1,\ldots,u_n\in \sig^+$, we denote this $u_1 \tok \cdots \tok u_n$. By $\sig^{\tok}$ we denote the set of all tokenizations constructed from strings from $\sig^*$. We refer to those strings as the \emph{tokens}. 
Let $\pi : \sig^{\tok} \to \sig^*$ be the concatenation of the strings in a tokenization, e.g.\ $\pi(u_1\tok \cdots \tok u_n)=u_1\cdots u_n\in \sig^*$. As a special case, we let $\pi$ applied to the tokenization with $n=0$, be the empty string. 
When $\pi(u_1\tok \cdots \tok u_n)=w\in\Sigma^*$, we say that \emph{$u_1\tok \cdots \tok u_n$ is a tokenization of $w$}.
For $\tau = u_1 \tok \cdots \tok u_n$, we denote by $|\tau|$ the integer $n$. Thus, $|\tau|=0$ if and only if $\tau$ is the empty tokenization.
Also, for $u\in\Sigma^*$, we let $|u|$ denote the length of $u$, i.e.\ the number of symbols from $\Sigma$ in the string $u$. It will be clear from the context and notation when $\tau$ denotes a tokenization with $|\tau|=1$, i.e.\ a tokenization of length one instead of a string of length one, given that in this case, $\tau$ could be interpreted as either. In fact, given our notational conventions discussed below, the symbol $\tau$ will always represent a tokenization.

In addition to using $|\tau|$ and $|u|$ for the length of a tokenization $\tau$ and length of a string $u$ respectively, we use $|S|$ to denote the cardinality of a (finite) set $S$. 

Differentiating between strings and tokenizations becomes important as we continue, so we adopt some conventions. Let $\sig$ denote the alphabet whenever not otherwise specified. When giving examples, we always use $\sig=\{a,b,c,\ldots\}$. Furthermore, we always let $\alpha,\beta,\gamma$ be variables denoting symbols from the alphabet, $u,v,w$ denote strings, and $\tau,\phi$ denote tokenizations, including all sub-/superscripted variants of each. In other words, $u,v,w\in \sig^*$ and $\tau,\phi,\psi\in \sig^{\tok}$, and for that matter $u_1,u_2\in \sig^*$, $\tau' \in \sig^{\tok}$, etc. As such we may write e.g.\ $u \tok \tau \tok v = \phi$ to mean that $\tau$ is a tokenization where the first token is $u$, the last is $v$, and the intervening tokens form the tokenization $\tau$, so $|\phi|=|\tau|+2$.

\section{Tokenizing Semantics}

First, we define a byte pair dictionary, which will be used to restrict the set of possible tokenizations for a given string $w$. We will only use $D$ and its sub-/superscripted variants to denote a dictionary.

\begin{definition}
    A byte pair \emph{dictionary $D=[u_1\tok v_1,\ldots,u_n\tok v_n]$} of length $|D|=n$ is a sequence of tokenizations $u_1\tok v_1,\ldots,u_n\tok v_n$, with each tokenization $u_i\tok v_i$ being of length 2. We call each $u_i \tok v_i$ a \emph{rule} and say that (a rule) $u_i \tok v_i$ has \emph{higher priority} than $u_j \tok v_j$, when $i<j$.
\end{definition}

We write $\tau \Rightarrow^D \tau'$ if $\tau=\phi \tok u \tok v \tok \phi'$ and $\tau'=\phi \tok uv \tok \phi'$, for some $u\tok v \in D$.
The dictionary $D$ will always be clear from the context, thus we omit the superscript on $\Rightarrow$. We let $\Rightarrow^+$ and $\Rightarrow^*$ denote the transitive and reflexive transitive closure of $\Rightarrow$, respectively. Also, for $w=\alpha_1\cdots\alpha_n$, we denote by $\TkEmpty(w)$ the tokenization $\alpha_1\tok\cdots\tok \alpha_n$. 

Next, we transfer terminology used in derivations over context-free grammars, to our setting. For $w=\alpha_1\ldots\alpha_n$, we begin a derivation for a tokenization with $\TkEmpty(w)$, although a more complicated pretokenizer step could certainly also be of interest, but not considered in this paper. Whereas in the case of context-free grammars, a derivation step consists of applying a grammar rule by replacing the non-terminal on the left-hand side of a rule, by its right-hand side, in our setting, a derivation step is of the form $\phi\tok u_i\tok v_i\tok\phi'\Rightarrow \phi\tok u_iv_i\tok\phi'$, for $u_i\tok v_i$ in $D$. A derivation terminates when no further rules from $D$ can be applied. 

The definition of a base tokenizer on $D$, which ignores the priority of rules in $D$, is as follows.

\begin{definition}
  \label{defn:tokenization}%
  For $D=[u_1 \tok v_1, \ldots, u_m \tok v_m]$ and $w=\alpha_1 \cdots \alpha_n$, we obtain the \emph{base tokenizations of $w$ by $D$,} denoted as $\Tk^D_{\textrm{base}}(w) \subset \sig^{\tok}$, as follows. 
We have $\tau_p \in \Tk^D_{\textrm{base}}(w)$ if: 
  \begin{itemize}
  \item $\tau_0=\TkEmpty(w)$,
    \item $\tau_0 \Rightarrow \cdots \Rightarrow \tau_p$,
    \item there exists no $\tau_{p+1}$ such that $\tau_p \Rightarrow \tau_{p+1}$.
  \end{itemize}
\end{definition}

That is, $\Tk^D_{\textrm{base}}(w)$ are the tokenizations of $w$ which can be achieved by applying rules to an initial tokenization $\TkEmpty(w)$ where all symbols in $w$ are their own token, until a point where no further rules can be applied. 
Thus, to obtain one of the possible base tokenizations, we select non-deterministically a rule from $D$ that can be applied to the current tokenization, until the set of rules that could be applied, is empty.
Given that $|\Tk^D_{\textrm{base}}|\ge1$, since there is non-deterministic choice in selecting the next applicable rule, a SentencePiece tokenizer~\cite{sentencepiece} is defined to remove ambiguity from the base tokenizer. We will (in a somewhat biased way) refer to this tokenization as the correct tokenization.

\begin{definition}
  \label{defn:tokenization-correct}%
  The \emph{SentencePiece} tokenization of $w$, denoted $\Tk^D(w)$, also referred to as the \emph{correct tokenization} of $w$, is $\tau_n$, with $\tau_n\in \Tk^D_{\textrm{base}}(w)$,  where:
  \begin{itemize}
  \item $\tau_0=\TkEmpty(w)$;
  \item $\tau_0 \Rightarrow \cdots \Rightarrow \tau_n$, and for $0\le i < n$, we pick the decomposition $\tau_i=\phi \tok u \tok v \tok \phi'$, to obtain $\tau_{i+1}=\phi \tok uv \tok \phi'$, in such a way that:
  \begin{itemize}
    \item $u\tok v$ is the highest priority rule in $D$ for which such a decomposition exists;
    \item among the remaining decompositions, we pick the \changeM{unique} one which minimizes $|\phi|$;
  \end{itemize} 
  \item no further rules apply to $\tau_n$.
  \end{itemize}
\end{definition}

Observe that $\Tk^D(w)$ always exists, and is obtained, intuitively, as follows.
Whenever it is possible to apply a rule from the dictionary $D$ to merge some tokens in the interim tokenization, merge the \emph{highest-priority} rule that occurs, and merge the left-most such pair if multiple occurrences exist. Note that this selects a unique tokenization, for each string $w$, from the set $\Tk^D_{base}(w)$.

\begin{example}
  Take the dictionary $D=[a\tok b, a \tok bc, b\tok c, ab \tok c]$, then the correct tokenization of the string $abcbcab$ is $abc \tok bc \tok ab$, using the following steps:
  \begin{itemize}
    \item Initially, $\tau_0 = a\tok b\tok c \tok b \tok c \tok a \tok b$, $a\tok b\in D$ applies to the leftmost $a \tok b$, producing $\tau_1 = ab \tok c \tok b \tok c \tok a \tok b$.
    \item The rule $a\tok b$ still applies, now to the last two tokens, producing $\tau_2 = ab \tok c \tok b \tok c \tok ab$. The rule $a \tok b$ now no longer applies anywhere.
    \item The next rule $a \tok bc$ does not apply, as there is no token $bc$, but $b \tok c$ does apply, producing $\tau_3 = ab \tok c \tok bc \tok ab$.
    \item The first rule that can now be applied is the rule $ab \tok c$, producing $\tau_4 = abc \tok bc \tok ab$. Now, no further rules apply, so $abc \tok bc \tok ab$ is the correct (SentencePiece) tokenization.
  \end{itemize}
  It is interesting to observe that the rule $a \tok bc$ has the second-highest priority in the dictionary, but was never applied. Also, when using the base tokenizer, $a \tok bc$ can certainly be applied when tokenizing the string $abcbcab$. Observe that applying $a\tok bc$ would produce the token $abc$, but tokenizing the string $abc$ takes the steps $a\tok b \tok c \Rightarrow ab \tok c \Rightarrow abc$. Later on in Corollary~\ref{cor:useless} we show that a rule $u_i\tok v_i$ is \emph{useful}, i.e.\ gets applied in the tokenization of some string, if and only if it gets applied when tokenizing the string $u_iv_i$. 
  \eqed
\end{example}

\begin{example}
  \label{ex:infinite-ripple}%
  Take the dictionary $D=[c\tok ab, ab \tok c, a \tok b]$. Then the correct tokenization of $abcabcabcabc$ is $abc \tok abc \tok abc \tok abc$. Notice that this tokenization is achieved left to right. After five steps, we have the tokenization $abc \tok abc \tok ab \tok c \tok a \tok b \tok c$. Contrast this to tokenizing the string $bcabcabcabc$ (i.e.\ we delete the initial $a$) which tokenizes as $b\tok cab \tok cab \tok cab\tok c$, or $cabcabcabcabc$ (i.e.\ adding an initial $c$), which leads to $cab \tok cab \tok cab \tok cab \tok c$, where we need to move left to apply $c\tok ab$, after having applied $a \tok b$. This illustrates that a small modification to the string can cause an arbitrarily large change to the resulting tokenization (to the right).
  \eqed
\end{example}

Interestingly enough, not all tokenization libraries modify the base tokenizer in the same way in order to eliminate ambiguity of tokenization. Let us consider the Python implementation of the GPT-2 tokenizer offered by HuggingFace~\cite{huggingface-gpt-2-py}, which removes ambiguity from the base tokenizer, as follows.

\begin{definition}
  \label{defn:tokenization-hf}%
  The \emph{HuggingFace tokenization} $\tau_n$ of $w$, which we denote by $\Tk^D_{\textrm{hf}}(w)\in \Tk^D_{\textrm{base}}(w)$, is defined as follows, where $\tau_0 = \TkEmpty(w)$. In $\tau_0 \Rightarrow^+ \cdots \Rightarrow^+ \tau_n$, for $0\le i < n$, we select a decomposition $\tau_i=\phi \tok u \tok v \tok \phi'$, such that $u\tok v\in D$ is the highest priority rule applicable to $\tau_i$, and then apply $u\tok v$ from left to right, until it is no longer applicable, in order to obtain \changeM{the unique} $\tau_{i+1}$.
\end{definition}

That is, in both Definitions~\ref{defn:tokenization-correct} and~\ref{defn:tokenization-hf} we rewrite $\tau_i$ by picking the highest-priority applicable rule and applying it at the left-most possible position. However, the SentencePiece semantics picks the highest-priority applicable rule in \emph{every} step, where HuggingFace picks a rule and uses it until it becomes inapplicable. This does create a formal difference in semantics, but as we will see they differ only in cases which may be considered degenerate given the way dictionaries are usually constructed.

\begin{example}
  Take the dictionary $D=[ab\tok a, a\tok b]$ and consider the tokenization of $w=abababab$, which has $\Tk^D(w)=aba \tok b \tok aba \tok b$, but $\Tk^D_{\textrm{hf}}(w)=ab \tok ab\tok ab \tok ab$.
  \eqed
\end{example}

\begin{remark}
  \label{remark:bpe-construction}%
The reason why the $D$ in the previous example is regarded to be degenerate or improper, is that a byte pair dictionary is typically produced~\cite{bpe} by taking a training corpus, initially tokenizing it symbol by symbol, and then iteratively adding the most common token pair to the dictionary, tokenizing, and repeating. For example, in the training corpus $a\tok b \tok c\tok a\tok b\tok c\tok a\tok c\tok a$, the most common token pair is $c\tok a$, which is inserted as the first rule in the dictionary. Then, we continue using $a\tok b \tok ca \tok b \tok ca \tok ca$. Now, the most common pair is $b\tok ca$, and this rule is added to the dictionary, which now consists of the rules $[c\tok a, b \tok ca]$. The new dictionary now produces $a\tok bca \tok bca \tok ca$, as tokenization, and so on. Observe that when constructing a dictionary in this way, a rule $u_i\tok v_i$ cannot have higher priority than the rules needed to produce $u_i$ and $v_i$, a property formalized in the next definition. 
Also, when tokenizing the training corpus with the dictionary obtained through training, using HuggingFace semantics, the tokenization obtained will be the tokenization of the training corpus at the end of the training process, and each rule in the dictionary will be used in this tokenization.
\eqed
\end{remark}
\begin{definition}
  A dictionary $D=[u_1\tok v_1,\ldots u_i\tok v_i,\ldots,u_j\tok v_j,\ldots,u_n\tok v_n]$ is \emph{proper} if for each $j$ with $|u_j|>1$, there exists $i<j$ such that $u_j=u_{i}v_{i}$, and similarly, for each $j'$ with $|v_{j'}|>1$, there exists $i'<j'$ such that $v_{j'}=u_{i'}v_{i'}$.
\end{definition}

Note that a proper dictionary may still contain rules which are not useful. Consider for example the dictionary $D=[b\tok c,a\tok b,c\tok d, ab\tok cd]$. Then $D$ is proper, but $ab\tok cd$ is not useful. This can be seen by noting that when $\TkEmpty(w)$ equals $a\tok b\tok c \tok d$, then the first rule gets applied to produce $a\tok bc \tok d$, and thus it is not possible to apply the rules $a\tok b$ and $c\tok d$, in order so that $ab\tok cd$ could be applied. But note that if $D$ is constructed from a training corpus, then $D$ is proper and each rule in $D$ is useful. 

With the additional assumption that the dictionary is proper, the SentencePiece and HuggingFace tokenizers turn out to have equivalent semantics. This should to some extent be expected, as they are intended to achieve the same results, the HuggingFace approach effectively being a small simplification.
\begin{lemma}
  \label{lemma:proper-unifies-semantics}%
  If $D$ is proper, we have $\Tk^D(w)=\Tk^D_{\textrm{hf}}(w)$ for all $w$.
\end{lemma}
\begin{proof}
  By contradiction, assume that some $w$ has $\Tk^D(w)\ne \Tk^D_{\textrm{hf}}(w)$. Let $\tau_0 \Rightarrow \cdots \Rightarrow \tau_n$ be the tokenization steps taken by $\Tk^D(w)$, and $\phi_0 \Rightarrow \cdots \Rightarrow \phi_m$ the tokenization steps taken by $\Tk^D_{\textrm{hf}}$. Let $i$ be the smallest index such that $\tau_{i} \ne \phi_{i}$. Note, such an $i$ must exist, otherwise, one sequence would be a subsequence of the other, which is impossible by Definition~\ref{defn:tokenization}.
  Thus, one sequence cannot be a proper prefix of the other. We also have $i\ge 2$, as the semantics differ only in that Definition~\ref{defn:tokenization-hf} prefers repeating the previous rule over the highest priority one, but this difference can only be exhibited when there is a previous step to repeat.

  We then have $\tau_{i-2}=\phi_{i-2}$, $\tau_{i-1}=\phi_{i-1}$ and $\tau_i\ne \phi_i$. Let $r_0$ be the rule applied in $\tau_{i-2}\Rightarrow \tau_{i-1}$ (and $\phi_{i-2}\Rightarrow \phi_{i-1}$ as they are equal), $r_1$ the rule in $\tau_{i-1} \Rightarrow \tau_i$, and $r_2$ the rule in $\tau_{i-1} \Rightarrow \phi_i$. That is, with some abuse of notation, the following situation:
  \begin{center}
    \begin{tikzpicture}[node distance=0pt]
      \node (a) {$\cdots \Rightarrow (\tau_{i-2}=\phi_{i-2}) \xRightarrow{r_0} (\tau_{i-1}=\phi_{i-1})$};
      \node[right=of a,yshift=6pt,rotate=15] {$\xRightarrow{r_1} \tau_i \Rightarrow \cdots$};
      \node[right=of a,yshift=-6pt,rotate=-15] {$\xRightarrow{r_2} \phi_i \Rightarrow \cdots$};

    \end{tikzpicture}
  \end{center}
  Then we know that $r_1$ has higher priority than $r_2$, since they must differ, and Definition~\ref{defn:tokenization-correct} (SentencePiece) always picks the highest priority rule applicable. The only possible reason for them to differ is that $r_0=r_2$, i.e.\  the Definition~\ref{defn:tokenization-hf} (HuggingFace) semantics prioritized using the same rule as in the previous step. However, as they agree in the previous step this means that we picked $r_0$ in that step by virtue of Definition~\ref{defn:tokenization-correct}, even though $r_1$ has higher priority than $r_0$, which must mean that $r_1$ was \emph{not} applicable before. This leads to a contradiction, as applying $r_0$ must then have created a token which made $r_1$ applicable, which, since $r_0$ \changeB{is} of lower priority, contradicts $D$ being a proper dictionary. As such, our assumption was wrong and $\Tk^D_{\textrm{hf}}(w)=\Tk^D(w)$ by necessity. 
\end{proof}

With this result in hand, it becomes less relevant to differentiate between the two semantics whenever considering only proper dictionaries. 

\begin{remark}
  It can be decided whether $D$ is proper in time $\mathcal{O}(\|D\|^2)$ (assuming $|uv|$ is constant for all rules $u\tok v$ in $D$), where $\|D\|=\sum \{|uv| \mid u\tok v \in D\}$. 
  For each $u\tok v\in D$, determine all $u'\tok v'$ such that $uv$ is a substring of $u'$ or $v'$, and for all such rules  $u'\tok v'$, verify that $u\tok v$ has lower priority than $u'\tok v'$.
 \eqed
\end{remark}


Next, we investigate the relationship between the tokenization of substrings of $w$, and the tokenization of $w$. First, we consider the following example. Let $u,v$ be strings with $\Tk^D(u)=\tau_1 \tok \tau_2$ and $\Tk^D(v)=\phi_1 \tok \phi_2$. Then it is not necessarily the case that $\tau_1 \tok \phi_2 = \Tk^D(\pi(\tau_1\tok \phi_2))$ or that $\phi_1 \tok \tau_2 = \Tk^D(\pi(\phi_1\tok \tau_2))$. An easy counterexample is obtained by letting $D=[a \tok a, b\tok b]$, $\tau_1 = a$, $\tau_2=b$, $\phi_1=b$, and $\phi_2=a$. Observe that indeed $\Tk^D(ab)=a\tok b= \tau_1 \tok \tau_2$, and $\Tk^D(ba)=b\tok a=\phi_1 \tok \phi_2$, \emph{however} $\tau_1\tok \phi_2=a\tok a\ne \Tk^D(aa)$ and $\phi_1 \tok \tau_2 = b\tok b \ne \Tk^D(bb)$. 
This shows that tokenizations can not be decomposed and then again glued together in arbitrary ways. 
However, deriving the tokenization of substrings of a given string $w$, given the final full tokenization of $w$, is sometimes possible, as shown in the following lemma. 


\begin{lemma}
  \label{lemma:robust-tok}%
  Tokenization derivations and tokenizations have the following properties:
  \begin{enumerate}
  \item[(i)] For both SentencePiece and HuggingFace, if $\phi_1\tok\ldots\tok\phi_k\Rightarrow^*\phi'_1\tok\ldots\tok\phi'_k$, then if $\pi(\phi_i)=\pi(\phi'_i)$ for all $i$, we have that $\phi_i\Rightarrow^*\phi'_i$ for all  $i$. 
  \item[(ii)] For a dictionary $D$ and string $w$ such that $\Tk^D(w)=\tau_1 \tok \ldots \tok \tau_k$ (or $\Tk_{hf}^D(w)=\tau_1 \tok \ldots \tok \tau_k$), it holds that $\Tk^D(\pi(\tau_i))=\tau_i$ (respectively $\Tk_{hf}^D(\pi(\tau_i))=\tau_i$). 
  \end{enumerate}

\end{lemma}
\begin{proof}
  For (i), let $\phi_{1,1} \tok \ldots \tok \phi_{k,1} \Rightarrow \cdots \Rightarrow \phi_{1,n} \tok \ldots \tok \phi_{k,n}$ be the steps taken by the procedure in Definition~\ref{defn:tokenization}, such that $\phi_{i,1}=\phi_i$ and $\phi_{i,n}=\phi'_i$ for all $i$, and $\pi(\phi_{i,j})=\pi(\phi_{i,j'})$ for all $i,j$ and $j'$.
  Removing all duplicates from $\phi_{i,1},\ldots,\phi_{i,n}$ produces the sequence of steps taken by SentencePiece or HuggingFace derivations, i.e.\ 
  in each step we apply the highest-priority rule from $D$ as left-most as possible, in the case of SentencePiece semantics, and we apply the highest-priority rule  as many times as possible, in the case of HuggingFace semantics.

  For (ii), take $\phi_i=\TkEmpty(\tau_i)$ and $\phi'_i=\tau_i$ in (i).
\end{proof}

\begin{remark}
  \label{rem:parts-of-toks}%
  A trivial outcome of this lemma is then that one can freely truncate tokenizations. I.e.\ if we have tokenized a long text $w$ and are only interested in a prefix, we can pick a suitable prefix of the tokenization (not the string) and it will be correct for the string it represents. A similar remark holds for a suffix of a tokenization.
\end{remark}




\begin{corollary}\label{cor:useless}
For SentencePiece or HuggingFace semantics, a rule in a dictionary $D$ is useful if and only if it gets applied when tokenizing the string it produces.  
\end{corollary}
    
\begin{proof}
The ``if'' part follows directly from the definition of useful, so for the converse, assume
we have a derivation $\alpha_1\tok\cdots\tok\alpha_n\Rightarrow \cdots \Rightarrow \phi \tok u \tok v \tok \phi' \Rightarrow \phi\tok uv\tok \phi'$. But then the previous lemma implies that $\TkEmpty(uv)\Rightarrow^* u\tok v\Rightarrow uv$.
\end{proof}

The main purpose of Lemma~\ref{lemma:robust-tok} is that it makes it possible to do tokenizations in streaming and incremental ways. 
Algorithm~\ref{algo:incremental-concatenate} below shows how this is achieved, but in order to establish the correctness of this algorithm, we first need the following corollary, which shows how we can split and then glue tokenizations.
\begin{corollary}
  \label{cor:bordered-stays-same}%
  If $\Tk^D(v)=\tau_1 \tok u \tok \tau_2$ and $\Tk^D(w)=\phi_1 \tok u \tok \phi_2$ then we also have $\Tk^D(\pi(\tau_1 \tok u \tok \phi_2))=\tau_1 \tok u \tok \phi_2$ and $\Tk^D(\pi(\phi_1 \tok u \tok \tau_2))=\phi_1 \tok u \tok \tau_2$. The same result holds for the HuggingFace semantics.
\end{corollary}
\begin{proof}
Remark~\ref{rem:parts-of-toks} give us that $\Tk^D(\pi(\tau_1\tok u))=\tau_1\tok u, \Tk^D(\pi(u\tok \tau_2))=u\tok\tau_2$ and similarly $\Tk^D(\pi(\phi_1\tok u))=\phi_1\tok u, \Tk^D(\pi(u\tok \phi_2))=u\tok\phi_2$ where $u$, $\tau_1$, $\phi_1$, $\tau_1$, and $\tau_2$ are as in the corollary above. 
The result now follows from the observation that we can glue two tokenizations together, if the end token of the first tokenization is the same as the start token of the next tokenization. The same argument holds for HuggingFace semantics. This follows the same line of argument as Lemma~\ref{lemma:robust-tok}, except instead of pruning steps from one tokenization, we interleave the steps of two, deduplicating rules applied to the overlapping token. The overlapping token ensures that one tokenization cannot ``disturb'' the other.
\end{proof}

With this in hand we can define an incremental update algorithm, which is not necessarily efficient in general (due to cases like Example~\ref{ex:infinite-ripple}), but will often do much less work than full retokenization, if we are in a situation where steps 4 and 5 are performed only a few times, i.e.\ if $i$ is close to $n$ and $j$ close to $1$ when the condition `If ($v_1=u_i$ or $i=1$) and ($v_k=u_j'$ or $j=m$)' in step 3 holds. When showing the correctness of Algorithm~\ref{algo:incremental-concatenate} below, we will use the special case of the previous corollary where $\tau_1$ and $\phi_2$ are empty tokenizations, or $\tau_2$ and $\phi_1$ are empty.


\begin{algorithm}
  \label{algo:incremental-concatenate}%
  Given $\Tk^D(w)=\tau$ and $\Tk^D(w')=\tau'$ we compute the tokenization $\Tk^D(ww')$ in the following way, assuming we are not in the trivial case where $w = \epsilon$ or $w' = \epsilon$.
  \begin{enumerate}
    \item Let $\tau=u_1\tok \cdots \tok u_n$ and $\tau'=u_1' \tok \cdots \tok u_m'$, initialize $i=n$ and $j=1$.
    \item Compute $\Tk^D(u_i \cdots u_n u_1' \cdots u_j')=v_1 \tok \cdots \tok v_k$.
    \item If ($v_1=u_i$ or $i=1$) and ($v_k=u_j'$ or $j=m$) output $u_1\tok \cdots \tok u_i \tok v_2 \tok \cdots \tok v_{k-1} \tok u_{j}' \tok \cdots \tok u_m'$ as $\Tk^D(ww')$ and halt.
    \item If $u_i\ne v_1$ and $i>1$, then $i\leftarrow i-1$.
    \item If $u_j'\ne v_k$ and $j<m$, then $j\leftarrow j+1$.
    \item Go to step 2.
  \end{enumerate}
\end{algorithm}

\begin{theorem}
  Algorithm~\ref{algo:incremental-concatenate} is correct.
\end{theorem}
\begin{proof}
  This amounts to two applications of Corollary~\ref{cor:bordered-stays-same}. The algorithm halts in a state where both $u_1 \tok \cdots \tok u_i$ and $v_1\tok \cdots \tok v_k$ are correct tokenizations, the former holds by Remark~\ref{rem:parts-of-toks} and the latter by construction. We also have $u_i=v_1$ (either that or $i=1$, but that case is trivial), which allows the application of Corollary~\ref{cor:bordered-stays-same} to establish that $\tau_1\tok u_i \tok v_2 \tok \cdots \tok v_k$ is a correct tokenization. To make this specific, in the terms of Corollary~\ref{cor:bordered-stays-same} we have $\tau_2=\phi_1=\eps$, $u=u_i=v_1$, $\tau_1=u_1 \tok \dots \tok u_{i-1}$, and $\phi_2=v_2 \tok \cdots \tok v_k$, which gives us that $\tau_1 \tok u \tok \phi_2=u_1\tok \cdots u_i \tok v_2 \tok \cdots \tok v_k$ is a correct tokenization.
  Now repeat this argument for the suffix to complete the proof.
\end{proof}

Remark~\ref{rem:parts-of-toks} and Algorithm~\ref{algo:incremental-concatenate} give tools to perform arbitrary incremental updates. For example, assume we have the tokenization $\tau$ of a (long) string $w$, and we make a small change in $w$, let's say $w=v_1 \alpha v_2$ and the updated string is $w'=v_1 \beta v_2$. Then let $\tau_1$ be the prefix of $\tau$ which falls entirely within $v_1$, let $\tau_2$ be the suffix of $\tau$ which falls entirely inside $v_2$, let $u$ be the string such that $\pi(\tau_1)u\pi(\tau_2)=w'$, then compute $\phi=\Tk^D(u)$, and obtain the tokenization of $w'$ by applying Algorithm~\ref{algo:incremental-concatenate} to concatenate $\tau_1$ to $\phi$, and then that tokenization to $\tau_2$.

Unfortunately, Algorithm~\ref{algo:incremental-concatenate} is not \emph{necessarily} more efficient than retokenizing the whole string, and might be worse. 
Of course, an experimental bound can be used in steps 3 and 4, where once we have decreased $i$ and increased $j$, more times than the specified bound, we fall back to retokenizing the complete string. In cases where there is a subset of symbols $\Sigma'$ from $\Sigma$, where symbols from $\Sigma'$ can only appear as the first or last symbols in $uv$ for rules $u\tok v$, and in all input strings there is a relatively small distance between symbols in $\Sigma'$ (i.e.\ input is not selected from all of $\Sigma^*$), we certainly have that Algorithm~\ref{algo:incremental-concatenate} is much more efficient than retokenizing the whole string.
In general, it seems likely that steps 4 and 5 in Algorithm~\ref{algo:incremental-concatenate} will be performed relatively few times in most practical dictionaries, but investigating this, is left for future work. 
We consider the worst-case in more generality in the next section. We establish a bound on how many times $i$ can be decremented determined solely by the dictionary (so a constant in the length of the string). Thus, we consider the worst-case when modifications happen late in a string, for example when applying appends. 

\section{Tokenizing Online with Finite Lookahead}

In this section, we assume all dictionaries are proper (although at times we do state this explicitly, to emphasize that we are making this assumption), so by Lemma~\ref{lemma:proper-unifies-semantics} $\Tk^D$ and $\Tk^D_{\textrm{hf}}$ are interchangeable. We investigate the following question: When we tokenize a string, in a streaming fashion, how long is the suffix that we need to tokenize again, when we resume tokenization, given we make no assumptions about the input string being tokenized. We refer to this constant as the \emph{lookahead constant}
 for a dictionary $D$, and denote it by $l(D)$. More formally, we have the following definition.
 
 \begin{definition}
  \label{defn:lookahead}%
Let $D$ be proper and $\phi, \tau$ and $\psi$ be tokenizations. Then $l(D)$, the \emph{lookahead constant} for $D$, is the smallest constant such that if $|\pi(\tau)|\ge l(D)$ and $\Tk^D(\pi(\phi)\pi(\tau))=\phi\tok \tau$, then $\Tk^D(\pi(\phi \tok \tau) \pi(\psi)) = \phi \tok \psi'$, for some tokenization $\psi'$.
 \end{definition}
 
 We show in Theorem~\ref{thm:left-ripple-bound}, that $l(D)\le |D|\cdot\max\{|uv|\mid u\tok v\in D\}$, where $|D|$  is the number of rules in $D$. After this, in Remark~\ref{rem:improved-lookahead}, we explore how to improve on this bound.

 The next lemma will be used in Theorem~\ref{thm:left-ripple-bound} to perform HuggingFace tokenization in a straightforward inductive way, where for a proper dictionary $D=[u_1\tok v_1,\ldots, u_n\tok v_n]$ we can tokenize a string by first applying the rule $u_1 \tok v_1$ as many times as possible, then the rule $u_2\tok v_2$  as many times as possible, and so on.

\begin{lemma}
  \label{lemma:hf-priority-order}%
  Let $D$ be proper dictionary and $\Tk^D_{\textrm{hf}}(w)=\tau$. Assume $r_1, \ldots, r_n$ is the sequence of rules applied to produce $\tau$ according to Definition~\ref{defn:tokenization-hf} (i.e.\ using HuggingFace semantics). Then it must be the case that $r_1, \ldots, r_n$ are in order of decreasing priority.
\end{lemma}
\begin{proof}
  We proceed in a way similar to Lemma~\ref{lemma:proper-unifies-semantics}. By contradiction, assume that there exists some $r_i$ such that $r_{i+1}$ is of higher priority than $r_i$. This means that the tokenization after the first $i$ steps contains the pair $r_{i+1}= u \tok v$, but this pair cannot have been created by the applications of $r_i$, as it is of lower priority than $r_{i+1}$ and $D$ is proper, and it also cannot have existed in the tokenization when $r_i$ was picked as the rule to next apply, as that contradicts how rules are picked in Definition~\ref{defn:tokenization-hf}. As such, our assumption was wrong, and $r_i$ is of higher priority or equal to $r_{i+1}$. 
\end{proof}

We will use the term \emph{refinement} for the way a tokenization is developed in this way, i.e.\ a \emph{tokenization $\tau$ is a refinement of $\tau'$}, if  $\tau'$ can be obtained from $\tau$ by applying $\pi$ to some of the subtokenizations in $\tau$. For example, $\phi\tok\phi'\tok\phi''$ is a refinement of $\phi\tok\pi(\phi')\tok\phi''$ and also of $\pi(\phi)\tok\pi(\phi')\tok\phi''$. We can also consider the opposite notion, i.e.\ a tokenization $\tau'$ is \emph{coarser} than $\tau$ , if $\tau$ is a refinement of $\tau'$. Thus, a (final) tokenization of a given string is obtained by using rules from $D$ to obtain coarser and coarser tokenizations.
Recall, $|D|$ denotes the number of rules in $D$.


\begin{theorem}
  \label{thm:left-ripple-bound}%
  Let $D$ be proper, $|\tau|\ge |D|$ and $\Tk^D(\pi(\phi)\pi(\tau))=\phi\tok \tau$. Then $\Tk^D(\pi(\phi \tok \tau) \pi(\psi)) = \phi \tok \psi'$, for some tokenization $\psi'$. That is, by Definition~\ref{defn:lookahead} we have $l(D)\le |D|\cdot \max \{|uv| \mid u\tok v \in D\}$.
\end{theorem}

\begin{proof}
This is easier to see using $\Tk^D_{\textrm{hf}}$, which is equivalent to $\Tk^D$ by Lemma~\ref{lemma:proper-unifies-semantics}.
Let $D$ be the dictionary $[u_1 \tok v_1, \ldots, u_n \tok v_n]$ and $D_i$ be the $i$-length prefix of $D$, i.e.\ $D_i=[u_1 \tok v_1,\ldots, u_i \tok v_i]$ for each $i$. We determine $\Tk^D_{\textrm{hf}}(w)$ as follows: First calculate $\Tk^{D_1}_{\textrm{hf}}(w)$, then assuming we know $\Tk^{D_i}_{\textrm{hf}}$, we apply $u_{i+1}\tok v_{i+1}$ to $\Tk^{D_i}_{\textrm{hf}}(w)$ wherever possible (working left to right) to obtain $\Tk^{D_{i+1}}_{\textrm{hf}}$. This procedure is correct by Lemma~\ref{lemma:hf-priority-order}.
We refer to $\Tk^{D_i}(w)$ as a tokenization at level $i$, and note that $\Tk^{D_i}(w)$ is a refinement of $\Tk^D(w)$.

Let $\phi\tok\tau = w_1\tok\ldots\tok w_k$ and $w=\pi(\phi \tok \tau) \pi(\psi)$. We show that the tokenization of $w$ at level $i$ is a refinement of a tokenization of the form $w_1\tok\ldots\tok w_{k-i}\tok \psi_i$, thus that $\Tk^{D_i}_{\textrm{hf}}(w)$ is a refinement of a tokenization of the form $w_1\tok\ldots \tok w_{k-i}\tok\psi_i$, for some tokenization $\psi_i$. 
More precisely, when we use \changeB{the dictionary $D_i$,} then tokenizing $\pi(\phi\tok\tau)\pi(\psi)$, \changeB{instead of only $\pi(\phi\tok\tau)$,} changes at most the rightmost $i$ tokens in $\phi\tok\tau$.
This implies that when $i=n$, we obtain that $\Tk^{D_n}_{\textrm{hf}}(w)$ is a tokenization of the form $w_1\tok\ldots\tok w_{k-n}\tok \psi_n$, i.e.\ not only a refinement of a tokenization of the given form.


First, we show that  $\Tk^{D_1}_{\textrm{hf}}(w)$ is a refinement of a tokenization of the form $w_1\tok\ldots \tok w_{k-1}\tok\psi_1$. Note that $u_1\tok v_1$ could \changeB{potentially be applied to the substring} $w_k\pi(\psi)$, when tokenizing $\pi(\phi\tok\tau)\pi(\psi)$, but $u_1\tok v_1$ is not applied across any of the $(k-1)$ \changeM{boundaries between tokens} in $w_1\tok\ldots\tok w_k$. This must be the case, otherwise $u_1 \tok v_1$ 
would have been applied \changeB{over some of these $(k-1)$ boundaries between the tokens $w_i$, for $1\le i\le k$}, when computing the tokenization $w_1 \tok \ldots \tok w_k$ with the full dictionary $D$. 

Moving on to the next rule in terms of priority, $u_2\tok v_2$, we repeat the argument we used for $u_1\tok v_1$. More precisely, $u_2\tok v_2$ could potentially be applied, one or more times, to the substring $w_{k-1}\pi(\psi_1)$, when tokenizing $\pi(w_1\tok\ldots\tok w_{k-1}\tok\psi_1)$, but $u_2\tok v_2$ is not applied across any of the $(k-2)$ tokenization boundaries in $w_1\tok\ldots\tok w_{k-1}$, otherwise it would have when $\pi(\phi\tok\tau)$ was tokenized using $D$. Thus, $\Tk^{D_2}_{\textrm{hf}}(w)$ is a refinement of a tokenization of the form $w_1\tok\ldots \tok w_{k-2}\tok\psi_2$.

We iterate this procedure for $i=3,\ldots,n$ to obtain the theorem. 
\end{proof}

\begin{example}
In this example, we consider the bound $|\tau|\ge|D|$, in the previous corollary. Fix a positive integer $n$ and let $D$ be a dictionary with the following $n$ rules:\changeB{$$a_n\tok a_{n+1},\ a_{n-1}\tok a_na_{n+1},\ a_{n-2}\tok a_{n-1}a_na_{n+1},\ \ldots,\ a_1\tok a_2\ldots a_{n+1}$$}
Also, let $\phi=a_0$, $\tau = a_1\tok a_2\tok\ldots\tok a_n$, and let $\psi=a_{n+1}$. Then, $\Tk^D(\pi(\phi\tok \tau)\pi(\psi))=a_0\tok a_1a_2\ldots a_{n+1}$, and we can not move any prefix of the tokenization of $\tau$ to $\phi$, otherwise we no longer have $\Tk^D(\pi(\phi \tok \tau) \pi(\psi)) = \phi \tok \psi'$, for some tokenization $\psi'$.
%
%
\eqed
\end{example}


Theorem~\ref{thm:left-ripple-bound} may at first appear quite abstract, but they demonstrate a fact that is very useful in practice: a finite lookahead is sufficient to tokenize a string from left to right.
\begin{definition}
  The \emph{sufficient lookahead} of a proper dictionary $D$ is $|D|\cdot\max\{|uv| \mid u\tok v \in D\}$. Observe that by Theorem~\ref{thm:left-ripple-bound} the sufficient lookahead is greater than or equal to $l(D)$.
\end{definition}
That is, for a proper dictionary $D$ and a string $w$, if we know that $\phi$ is a prefix of $\Tk^D(w)$ (beginning the process by taking $|\phi|=0$) we can compute a prefix $\phi\tok u$ by inspecting only the sufficient lookahead many next symbols. Observe that this lookahead length does not depend on $|w|$. This has potential to improve tokenization performance by cache locality (where e.g.~SentencePiece~\cite{sentencepiece} and HuggingFace~\cite{huggingface-gpt-2-py} access the string contents with random access), but also enables doing streaming tokenization using a constant amount of memory, for when the entire string is not available or impractical to hold in memory. One way to express this finite state tokenization approach is as a deterministic string-to-string transducer. First, to avoid special cases for the end of the string, let us define a simple normal form.
\begin{definition}
  For a dictionary $D$ over the alphabet $\sig$, let $k$ be the sufficient lookahead for $D$, assume $z\notin \sig$, then a string $v \in (\sig\cup\{z\})^*$ is the \emph{end-padding of $w$} if it is of the form $wzz\cdots z$ where there are a total of $k$ trailing $z$s.
\end{definition}
\begin{remark}
  Observe that if $v$ is $w$ end-padded, then we have $\Tk^D(v)=\Tk^D(w) \tok \phi$ where $\phi=z\tok \cdots \tok z$, since $D$ contains no rules involving $z$. This means that tokenizing $v$ from the left to right we can end the procedure the moment the lookahead consists only of $z$s, as that is the padding which will always tokenize to $\phi=z\tok \cdots \tok z$.
  \eqed
\end{remark}
This allows us to state a straightforward left-to-right tokenization algorithm without having special cases for when the 
\begin{algorithm}
  \label{algo:finite-state}%
  Let $D$ be a proper dictionary over the alphabet $\sig$ and $k$ its sufficient lookahead. Assume that $z\notin \sig$.

  Precompute $f : (\sig\cup\{z\})^k \to \sig^{\tok}$ such that for all $w\in \sig\cup\{z\}$ we have $f(w)=u$ where $\Tk^D(w)=u \tok \tau$ for some $\tau$. Observe that then $\Tk^D(ww')=u \tok \tau'$ for all $w'$ as well, by Theorem~\ref{thm:left-ripple-bound}.

  Then for any string $w$ let $v$ be its end-padding, we can then compute $\Tk^D(v)$ using the following steps.
  \begin{enumerate}
    \item Split $v=v'v''$ such that $|v'|=k$.
    \item If $v'=z\cdots z$, halt.
    \item Lookup $f(v')=u$, output $u$.
    \item Split $v=uu'$ (observe that $u$ must be a prefix of $v'$ and in turn $v$).
    \item Update $v$ to be $u'$, then go to 1.
  \end{enumerate}
\end{algorithm}

\begin{theorem}
  For any fixed proper $D$ and any string $w$, Algorithm~\ref{algo:finite-state} outputs $\Tk^D(w)$ in time $\mathcal{O}(|w|)$ and using $\mathcal{O}(1)$ taking $D$ to be fixed and the input string to be read only.
\end{theorem}
\begin{proof}
  Correctness follows from Theorem~\ref{thm:left-ripple-bound}, with the algorithm picking tokens based on a long enough prefix of the string (guaranteed to exceed $l(D)$) that it is guaranteed that any suffixes will still have the correct tokenization produce that same token.

  The time and space bounds are trivial noting that each iteration step uses up part of the string, and suitably implemented each step uses an amount of time and space bounded in $l(D)$, which is $\mathcal{O}(1)$. Reusing space, the indicated bounds are reached.
\end{proof}

A perhaps more natural presentation of this algorithm would be constructing a string-to-string transducer. This is not very complicated, the transducer would read $l(D)$ symbols and then, in much the same way as the precomputed $f$ in Algorithm~\ref{algo:finite-state}, output a token accordingly. However, the bound here provided is quite large, and a more sophisticated construction is likely needed.

The bound established by Theorem~\ref{thm:left-ripple-bound} is clearly quite loose for most realistic dictionaries. Basically, it assumes that the rules in the dictionary create a chain, where each successive rule can ``interfere'' with the application of the next lower priority rule.
\begin{remark}\label{rem:improved-lookahead}
  We have from Theorem~\ref{thm:left-ripple-bound} that $l(D)\le |D|\cdot\max\{|uv|\mid u\tok v\in D\}$, where $l(D)$ is the lookahead constant of $D$, since $|\pi(\tau)| \ge |D|\cdot\max\{|uv|\mid u\tok v\in D\}$, implies $|\tau|\ge |D|$.
  Next, we consider (informally) how to improve the bound $|\tau|\ge|D|$, which will then improve the bound on the lookahead constant. Consider dictionaries $D=[a\tok b, c\tok d, ab\tok cd]$ and $D'=[c\tok d, a\tok b, ab\tok cd]$, i.e.\ $D$ and $D'$ have the same rules, but we switched the order of the first two rules in these dictionaries.  Now,  note that the tokenizations of any string $w$ will be the same, independent of if we use $D$ or $D'$. We say that dictionaries $D$ and $D'$ are \emph{equivalent} if $\Tk^D(w)=\Tk^{D'}(w)$ for all strings $w\in\Sigma^*$. The complexity of deciding if two dictionaries are equivalent, and if equivalence is even  decidable, is left for future research.

  Next, we define the \emph{chain length} of a dictionary $D$, denoted as $c(D)$. For a dictionary $D$, we let $c(D)$ be the maximum  value of $n$ such that we have a sequence of rules of decreasing priority, $r_1,\ldots,r_n$, in all dictionaries equivalent to $D$. 
  Thus, we certainly have that $c(D)\le |D|$. Again, the complexity of computing $c(D)$ will be considered in a future publication, but once we have some, but necessary all dictionaries equivalent to $D$, we can obtain an upper bound for $c(D)$, most likely better than $|D|$. In particular, if $r=u\tok v$ and $r'=u'\tok v'$ are neighbouring rules in $D$, such that a non-empty suffix of $uv$ is not a prefix of $u'v'$, a non-empty prefix of $uv$ is not a suffix of $u'v'$, and $uv$ is not a substring of $u'v'$, and we have similar conditions when swapping $r$ and $r'$, then certainly we can switch $r$ and $r'$ in $D$ and this will not change the tokenization of any string.
  For example, if $D=[a\tok b, c\tok d, ab\tok cd]$, then $c(D)\le 2$, since $[a\tok b, c\tok d, ab\tok cd]$ and $[c\tok d, a\tok b, ab\tok cd]$ are equivalent dictionaries.

  Finally, note that with these concepts in hand, the proof of Theorem~\ref{thm:left-ripple-bound} actually shows a stronger result: the theorem holds if we replace the bound $|\tau|\le |D|$ by $|\tau|\le c(D)$. 
  This can be seen by noting that the proof of Theorem~\ref{thm:left-ripple-bound} implies that there is a sequence of rules $r_1,\ldots,r_n$, such that if we apply these rules in order, as in HuggingFace semantics, then we obtain a refinement of $w_1\tok\ldots\tok w_{k-i}\tok\psi_i$ after having applied $r_i$. But, independently of which dictionary equivalent to $D$ is used, $r_1$ is only potentially applied over the tokenization boundary between  $w_k$ and $\psi$, and not over any of the $(n-1)$ tokenization boundaries between any of the $w_i$. Similarly, none of the $r_i$, with $i\ge 2$, is applied over any of the $(k-i-1)$ tokenization boundaries in $w_1\tok\ldots\tok w_{k-i}$. This is the case, since with equivalent dictionaries, by definition, we obtain the same tokenization, and as shown in the proof of Theorem~\ref{thm:left-ripple-bound}, in order to cross the next tokenization boundary from the right, we need a lower priority rule. In summary, each $r_i$ from $r_1\ldots r_n$, will, in order, cross at most one more tokenization boundary, from the right, in $w_1\tok\ldots\tok w_{k}\tok\psi_i$, and each $r_i$ has lower priority than $r_{i-1}$, independent of which dictionary equivalent to $D$ is considered. 
  We thus have that $l(D)\le c(D)\cdot\max\{|uv|\mid u\tok v\in D\}$.

It is necessary to obtain $c(D)$ through a more efficient procedure than pure enumeration. 
Note, the enumeration involved in constructing $f$ in Algorithm~\ref{algo:finite-state} inefficiently ``finds'' the true $l(D)$ anyway.
  \eqed
\end{remark}

\section{Conclusions and Future Work}

In some ways, the main contribution of this paper is the more formal definition of the tokenization semantics, allowing them to be studied in closer details. We leveraged this to establish some interesting properties of tokenizations, and established algorithms for both incrementally modifying a tokenization and for doing tokenizations left-to-right using space constant in the length of the string.

Much future work remains, including the following.
\begin{itemize}
  \item Experimental studies should be performed, for example, testing how the incremental algorithm behaves in random cases. It seems likely to be extremely efficient in practice, as long chains of changes, or infinite ones such as in Example~\ref{ex:infinite-ripple}, do not seem to be very common or realistic. \changeM{We do not offer implementation details in this paper, as in \emph{general} the bounds implied by the constructions are high enough to be impractical.}
  \item Determining better upper bounds for the lookahead constant. Some more aspects that will be considered, are listed throughout the paper, and in particular in Remark~\ref{rem:improved-lookahead}. \changeM{Once such bounds are established practical implementation details can be considered.}
  \item Beyond improving the bounds of Theorem~\ref{thm:left-ripple-bound} there is also its converse, determining how much a tokenization may change from appending strings on the \emph{left.} Example~\ref{ex:infinite-ripple} already demonstrates that this is not finite for all $D$, but from random testing it seems to \emph{often} be finite. It should be investigated whether the dictionaries exhibiting these infinite ripples of changes do so due to some easily decidable property, and whether the ``lookbehind'' (compare Definition~\ref{defn:lookahead}) is efficiently computable when finite.
\end{itemize}


\bibliographystyle{eptcs}
\bibliography{main}

\end{document}